%% file: BTWSC.tex
\newif\ifprocs
\procstrue
\procsfalse %last one wins

\ifprocs
\documentclass[oribibl]{llncs}
\else
\documentclass[letterpaper, 11pt]{article}
\usepackage{fullpage,amsthm}
\fi

\usepackage{latexsym,amssymb,amsmath}
\usepackage{algorithm, algorithmic}
\newcommand{\remove}[1]{}
\usepackage{xspace}
\usepackage{color}
\usepackage{multirow}
\usepackage[hypertex]{hyperref}

\def\compactify{\itemsep=0pt \topsep=0pt \partopsep=0pt \parsep=0pt}

\newcommand{\mnote}[2]
{\mbox{}\marginpar%
[{\tiny\it\begin{minipage}[t]{\marginparwidth}\raggedleft#1%
\end{minipage}}]%
{\tiny\it\begin{minipage}[t]{\marginparwidth}\raggedright#1%
\end{minipage}}%
{}#2}

\ifprocs
\newcommand{\pparagraph}[1]{\subsubsection*{{#1}}}
\else
\newcommand{\pparagraph}[1]{\paragraph{{#1}}}
\fi

\ifprocs\else
\newtheorem{theorem}{Theorem}[section]
\newtheorem{lemma}[theorem]{Lemma}

\newtheorem{corollary}[theorem]{Corollary}

\newtheorem{claim}[theorem]{Claim}

\theoremstyle{remark}
\newtheorem{remark}[theorem]{Remark}

\theoremstyle{definition}
\newtheorem{definition}[theorem]{Definition}
\fi

\newcommand{\Real}{\mathbb R}
\newcommand{\expec}{{\mathbb{E}}}

\newcommand{\eps}{\varepsilon}

\newcommand{\prob}{{\rm Pr}}
\newcommand{\var}{{\rm Var}}

\newcommand{\poly}{{\rm poly}}
\def\SC{{\sf Sparsest-Cut}\xspace}
\def\tO{{\tilde O}}

\DeclareMathOperator{\dem}{{\sf dem}}
\DeclareMathOperator{\capp}{{\sf cap}}

\DeclareMathOperator{\tw}{tw}

\def\SA{\textbf{SA}}
\def\SCLP{\textbf{SC}}
\def\SCRound{\textsf{SC-Round}}
\newcommand{\distr}{\mu}

\newtheorem{bags-lemma}{Lemma~\ref{lem:bags}}
\newtheorem{order-lemma}{Lemma~\ref{lem:order-invariant}}
\newtheorem{stochastic-lemma}{Lemma~\ref{lem:A-stochastic}}
\newtheorem{potential-lemma}{Lemma~\ref{lem:potential}}

\begin{document}

\title{Approximating Sparsest Cut in \ifprocs\\ \else\fi Graphs of Bounded Treewidth%
\ifprocs
\thanks{A full version appears at \url{http://arxiv.org/abs/1006.3970}}
\else
\fi
}

\ifprocs
\author{Eden Chlamtac\inst{1}\thanks{Supported in part by a Sir Charles Clore postdoctoral fellowship.}
\and
Robert Krauthgamer\inst{1}\thanks{Supported in part by The Israel Science Foundation (grant \#452/08), and by a Minerva grant.}
%\footnotemark[1]
\and
Prasad Raghavendra\inst{2}
}

\institute{Weizmann Institute of Science, Rehovot, Israel.
    \email{\{eden.chlamtac,robert.krauthgamer\}@weizmann.ac.il}
\and
Microsoft Research New England, Cambridge, MA, USA.
    \email{pnagaraj@microsoft.com}
}

\else
\author{Eden Chlamtac%
  \thanks{This work was supported in part by The Israel Science Foundation
    (grant \#452/08), and by a Minerva grant. The first author was supported by a Sir Charles Clore postdoctoral fellowship.
    Weizmann Institute of Science, Rehovot, Israel.
    Email: \texttt{\{eden.chlamtac,robert.krauthgamer\}@weizmann.ac.il}
  }
  \\Weizmann Institute
\and Robert Krauthgamer\footnotemark[1]
  \\Weizmann Institute
\and Prasad Raghavendra\thanks{Email: \texttt{pnagaraj@microsoft.com}
} \\ Microsoft Research New England\\
%Cambridge, MA, USA.
}
\fi

%\date{}

\ifprocs\else
\date{}
%\begin{titlepage}
\fi

\maketitle

\begin{abstract}
\input{abstract}
\end{abstract}

\ifprocs\else
%\thispagestyle{empty}
%\end{titlepage}
\fi

\section{Introduction}
\input{intro.tex}

\section{Technical Overview}
\input{overview.tex}

\section{The Algorithm}
\input{algo.tex}

\section{Markov Flow Graphs}\label{sec:flows}
\input{flow.tex}

\ifprocs\else
\section{Bounding the Cut Size}
\input{cut.tex}

\input{mainproof.tex}

\section{A Lower Bound for the Rounding Algorithm}\label{sec:lb}
\input{lb.tex}

\fi

\pparagraph{Acknowledgments}
We would like to thank Claire Mathieu for a series of helpful conversations.

{\small
\ifprocs
\bibliographystyle{splncs03}
\else
\bibliographystyle{alphainit}
\fi
\bibliography{robi,btwsc,drafts}
}

\ifprocs\else
\appendix

\input{hardness}
\fi

\end{document}

%% file: abstract.tex
We give the first constant-factor approximation algorithm for \SC with general demands in bounded treewidth graphs. 
In contrast to previous algorithms, which rely on the flow-cut gap and/or metric embeddings, our approach exploits the Sherali-Adams hierarchy of linear programming relaxations.

%% file: intro.tex
%\section*{Introduction}

% sparsest-cut problem
The \SC problem is one of the most famous graph optimization problems.
The problem has been studied extensively due to the central role it plays in several respects.
First, it represents a basic graph partitioning task
that arises in several contexts,
such as divide-and-conquer graph algorithms
(see e.g.~\cite{LR99,Shmoys:CutSurvey} and~\cite[Chapter 21]{Vazirani01}).
Second, it is intimately related to other graph parameters,
such as flows, edge-expansion, conductance, spectral gap and bisection-width.
Third, there are several deep technical links between \SC
and two seemingly unrelated concepts,
the Unique Games Conjecture and Metric Embeddings.

%Unlike the classical minimum $st$-cut problem,
Given that \SC is known to be NP-hard~\cite{MS86},
the problem has been studied extensively from the perspective of
polynomial-time approximation algorithms.
Despite significant efforts and progress in the last two decades,
we are still quite far from determining the approximability of \SC.
This is true not only for general graphs,
but also for several important graph families,
such as planar graphs or bounded treewidth graphs.
% even for graphs of treewidth $3$ (which appears rather restricted).
The latter family is the focus of this paper; we shall return to it
after setting up some notation and defining the problem formally.

% formal definition of the problem
\pparagraph{Problem definition.}
For a graph $G=(V,E)$ we let $n=|V|$.
For $S\subset V$, the cutset $(S,\bar S) \subset V \times V$
is the set of unordered pairs with exactly one endpoint in $S$,
i.e.\ $\{\{u,v\}\in V\times V:\ u\in S, v\notin S\}$.
In the \SC problem (with general demands), the input is
a graph $G=(V,E)$ with edge capacities $\capp:E\to\Real_{\ge 0}$ and
a set of \emph{demand pairs}, $D =
(\{s_1,t_1\},\ldots,\{s_k,t_k\})$ with a demand function $\dem : D
\to \Real_{\ge 0}$. The goal is to find $S\subset V$ (a cut of $G$) that minimizes the ratio
$$
  \Phi(S)
  = \frac{\sum_{(u,v)\in (S,\bar{S}) \cap E} \capp(u,v)} {\sum_{(u,v)\in
  (S,\bar{S}) \cap D} \dem(u,v)}.
$$
The demand function $\dem$ is often %derived from a collection
%$D=(\{s_1,t_1\},\ldots,\{s_k,t_k\})$ of \emph{demand-pairs},
set to %by setting
$\dem(s,t)=1$ for all $(s,t)\in D$. %
 %and $\dem(u,v)=0$ for all other $u,v$.
% uniform-demands version
The special case where, in addition to this, the demand set $D$ includes all vertex pairs %$u\neq v\in V$
 is referred to as \emph{uniform demands}.

\pparagraph{Treewidth.}
Let $G=(V,E)$ be a graph.
A \emph{tree decomposition} of $G=(V,E)$ is a pair $(\mathcal B,T)$
where ${\mathcal B}=\{B_1,\ldots,B_m\}$ is a family of subsets $B_i\subseteq V$ called \emph{bags},
and $T$ is a tree whose nodes are the bags $B_i$,
satisfying the following properties:
(i) $V=\bigcup_i B_i$;
(ii) For every edge $(u,v)\in E$, there is a bag $B_j$ that contains both $u,v$; and
(iii) For each $v\in V$, all the bags $B_i$ containing $v$
form a connected subtree of $T$.
The \emph{width} of the tree decomposition is $\max_i |B_i|-1$.
The \emph{treewidth} of $G$, denoted $\tw(G)$,
is the smallest width among all tree decompositions of $G$.
The \emph{pathwidth} of $G$ is defined similarly, except that
$T$ is restricted to be a path; thus, it is at least $\tw(G)$.
%hence the family of graphs of bounded treewidth is contains
%the family of bounded pathwidth graphs.
It is straightforward to see that every graph $G$
excludes as a minor the complete graph on $\tw(G)+2$ vertices.
Thus, the family of graphs of tree width $r$ contains the family of
graphs with pathwidth $r$, and is contained in the family of graphs
excluding $K_{r+2}$ as a minor (here $K_{r+2}$ refers to the complete
graph on $r+2$ vertices).

%\pparagraph{Bounded treewidth graphs.}
\subsection{Results}
We present the first algorithm for general demand \SC that achieves a constant factor approximation
for graphs of bounded treewidth $r$ (the restriction is only on the structure of the graph, not the demands).
Such an algorithm is conjectured to exist by~\cite{GNRS04} (they actually make
a stronger conjecture, see Section \ref{sec:related} for details).
However, previously such an algorithm was not known even for $r=3$,
although several %different
algorithms are known for %treewidth
 $r=2$~\cite{GNRS04,CJLV08,CSW10}
and for bounded-pathwidth graphs~\cite{LS09arxiv}
(which is a subfamily of bounded-treewidth graphs).

\begin{theorem}\label{thm:main}
There is an algorithm for \SC (general demands) on graphs of treewidth $r$,
that runs in time $(2^r n)^{O(1)}$ and achieves approximation factor $C=C(r)$
(independently of $n$, the size of the graph).
\end{theorem}

Table \ref{tab:UB} lists the best approximation algorithms known
for various special cases of \SC.
% We remark that the special case of uniform-demands and graphs of treewidth $r$
% can be solved exactly using dynamic programming in time $2^{O(t)} n^{O(1)}$;
% alternatively, an approximation $O(\log t)$ can be computed in time $n^{O(1)}$
% because there is such a bound on the flow-cut gap \cite{Rabinovich03,CKS09}.
We remark that the problem (with general demands) is NP-hard
even for pathwidth \ifprocs $2$ (see the full version for details).\else $2$;
we include a proof of this fact in Appendix~\ref{app:hardness} for the sake of completeness, as it is unclear whether this has appeared previously in the literature.
\fi

\begin{table}[htb]
  \begin{minipage}[f]{1.0\linewidth}
  \centering
  \begin{tabular}[t]{|l|lcll|}
  \hline
  Demands & Graphs & Approximation & Based on & Reference \\
  \hline
  \hline
    \multirow{6}{*}{
%      \begin{tabular}{@{}l@{}}
        general
%      \end{tabular}
    }
    & arbitrary & $\tO\big(\sqrt{\log |D|}\big)$ & SDP  & \cite{ALN08} \\
    & treewidth $2$ & $2$ & LP (flow) & \cite{GNRS04,CJLV08} \\
    & fixed outerplanarity & $O(1)$ & LP (integer flow) & \cite{CGNRS06,CSW10} \\
    & excluding $W_4$-minor & $O(1)$ & LP (flow) & \cite{CJLV08} \\
    & fixed pathwidth & $O(1)$ & LP (flow) & \cite{LS09arxiv} \\
%    & fixed treewidth & $O(1)^\dagger$ \footnotetext{$^\dagger$Runtime depends on the treewidth}
    & fixed treewidth & $O(1)$ & LP (lifted) & This work \\
  \hline
    \multirow{3}{*}{
%      \begin{tabular}{@{}l@{}}
        uniform
%      \end{tabular}
    }
    & arbitrary & $O(\sqrt{\log n})$ & SDP & \cite{ARV09} \\
    & excluding fixed-minor & $O(1)$ & LP (flow) & \cite{KPR93,FT03} \\
    & fixed treewidth & O(1) & LP (flow) & \cite{Rabinovich03,CKS09} \\
    & fixed treewidth & $1$ & \multicolumn{2}{l|}{dynamic programming} \\
  \hline
  \end{tabular}
  \end{minipage}
  \caption{Approximation algorithms for \SC.}
  \label{tab:UB}
\end{table}

\pparagraph{Techniques.}
%\label{sec:techniques}

Similarly to almost all previous work, our algorithm is based
on rounding a linear programming (LP) relaxation of the problem.
A unique feature of our algorithm is that it employs an LP relaxation
derived from the hierarchy of (increasingly stronger) LPs,
designed by Sherali and Adams~\cite{SA90}.
Specifically, we use level $r+O(1)$ of this hierarchy.
In contrast, all prior work on \SC uses either the standard LP
(that arises as the dual of the concurrent-flow problem, see e.g.~\cite{LR99}),
or its straightforward strengthening to a semidefinite program (SDP).
Consequently, the entire setup changes significantly (e.g. the known
connections to embeddings and flow, see Section~\ref{sec:conjecture}),
and we face the distinctive challenges of exploiting the
complex structure of these relaxations (see Section~\ref{sec:related}).

While bounding the integrality gap of the standard LP (the flow-cut gap) for various graph families remains an important open problem with implications in metric embeddings (see Section~\ref{sec:conjecture}), our focus is on directly approximating \SC. Accordingly, our LP is larger and (possibly much) stronger than
the standard flow LP, and
hence our rounding does not imply a bound on the flow-cut gap
(akin to rounding of the SDP relaxation in~\cite{ARV09,CGR08,ALN08}).

Finally, note that the running time stated in Theorem~\ref{thm:main} is much better than the $n^{O(r)}$ running time typically needed to solve the $r+O(1)$ level of Sherali-Adams (or any other hierarchy). The reason is that only $O(3^rn|D|)$ of the Sherali-Adams variables and constraints are really needed for our analysis to go through (see Remark~\ref{rem:LP-size}), thus greatly improving the time needed to solve the LP. As the rounding algorithm we use is a simple variant of the standard method of randomized rounding for LP's (adapted for Sherali-Adams relaxations on bounded-treewidth graphs), the entire algorithm is both efficient and easily implementable.

\subsection{The GNRS excluded-minor conjecture}
\label{sec:conjecture}

Gupta, Newman, Rabinovich and Sinclair (GNRS) conjectured in~\cite{GNRS04}
that metrics supported on graphs excluding a fixed minor embed
into $\ell_1$ with distortion $O(1)$ (i.e. independent of the graph size).
By the results of \cite{LLR95,AR98,GNRS04}, this conjecture is equivalent
to saying that in all such graphs (regardless of the capacities and demands),
the ratio between the sparsest-cut and the concurrent-flow,
called the \emph{flow-cut gap}, is bounded by $O(1)$.
Since the concurrent-flow problem is polynomial-time solvable
(e.g. by linear programming), the conjecture %being true
 would immediately
imply that \SC admits $O(1)$ approximation (in polynomial-time) on these graphs.

Despite extensive research, the GNRS conjecture is still open,
even in the special cases of planar graphs and of graphs of treewidth $3$.
The list of special cases that have been resolved includes graphs of treewidth $2$,
$O(1)$-outerplanar graphs, graphs excluding a $4$-wheel minor,
and bounded-pathwidth graphs;
see Table \ref{tab:UB}, where the flow LP is mentioned.

Our approximation algorithm may be interpreted as evidence supporting
the GNRS conjecture (for graphs of bounded treewidth),
since by the foregoing discussion,
the conjecture being true would imply the existence of such approximation algorithms,
and moreover that our LP's integrality gap is bounded. In fact, one consequence of our algorithm and its analysis can be directly phrased in the language of metric embeddings:

\begin{corollary} For every $r$ there is some constant $C=C(r)$ such that every shortest-path metric on a graph of treewidth $\leq r$, for which every set of size $r+3$ is isometrically embeddable into $L_1$ in a locally consistent way (i.e.\ the embeddings of two such sets, when viewed as probability distributions over cuts, are consistent on the intersection of the sets), can be embedded into $L_1$ with distortion at most $C$.
\end{corollary}

If, on the other hand, the GNRS conjecture is false,
then our algorithm (and its stronger LP) gives a substantial improvement
over techniques using the flow LP,
and may have surprising implications for the Sherali-Adams hierarchy (see Section~\ref{sec:related}). Either way, our result opens up several interesting questions,
which we discuss in Section \ref{sec:open}.

\subsection{Related work}
\label{sec:related}

\pparagraph{Relaxation hierarchies and approximation algorithms.}

A research plan that has attracted a lot of attention in recent years
is the use of lift-and-project methods to
design improved approximation algorithms for NP-hard optimization problems.
These methods, such as Sherali-Adams~\cite{SA90},
Lov{\'a}sz-Schrijver~\cite{LovaszSchrijver:Cones},
and Lasserre~\cite{Lasserre02} (see~\cite{Laurent} for a comparison),
systematically generate, for a given $\{0,1\}$ program
(which can capture many combinatorial optimization problems, e.g.\ {\sf Vertex-Cover}),
a sequence (aka \emph{hierarchy}) of increasingly stronger relaxations.
The first relaxation in this sequence is often a commonly-used LP relaxation
for that combinatorial problem. %, e.g. {\sf Vertex-Cover}.
After $n$ steps (which are often called \emph{rounds} or \emph{levels}),
the sequence converges to the convex hull of the integral solutions,
and the $k$-th relaxation in the sequence is a convex program (LP or SDP)
that can be solved in time $n^{O(k)}$.
Therefore, the first few, say $O(1)$, relaxations in the sequence
offer a great promise to approximation algorithms ---
they could be much stronger than the commonly-used LP relaxation,
yet are polynomial-time computable. This is particularly promising for problems for which there is a gap between known approximations and proven hardness of approximation (or when the hardness relies on weaker assumptions than $P\neq NP$).%, e.g.\ the Unique Games Conjecture).

Unfortunately, \ifprocs since \else starting with \fi the work of
Arora, Bollob{\'a}s, Lov{\'a}sz, and Tourlakis~\cite{ABLT06}
on {\sf Vertex-Cover}, there has been a long line of work showing that for various problems,
even after a large (super-constant) number of rounds, various hierarchies do not yield smaller integrality gaps than a basic LP/SDP relaxation (see, e.g.~\cite{STT07,GMPT07,Schoenebeck08,Tulsiani09,CMM09}). In particular, Raghavendra and Steurer~\cite{RS09} have recently shown that a superconstant number of rounds of certain SDP hierarchies does not improve the integrality gap for any constraint satisfaction problem (MAX-CSP).

In contrast, only few of the known results are positive,
i.e. show that certain hierarchies give a sequence of improvements in the integrality gap in their first $O(1)$ levels ---
this has been shown for {\sf Vertex-Cover} in planar graphs~\cite{MM09},
{\sf Max-Cut} in dense graphs~\cite{FK07},
{\sf Knapsack}~\cite{KarMatNgu,Bienstock08}, and
{\sf Maximum Matching}~\cite{MatSinc09}.
There are even fewer results where the improved approximation
is the state-of-the-art for the respective problem ---
such results include recent work on
{\sf Chromatic Number}~\cite{Chlamtac07},
{\sf Hypergraph Independent Set}~\cite{CS08}, and
{\sf MaxMin Allocation}~\cite{BCG09}.

In the context of bounded-treewidth graphs,
a bounded number of rounds in the Sherali-Adams hierarchy
is known to be tight (i.e. give exact solutions) for many problems that
are tractable on this graph family, such as CSPs~\cite{JordanWain}.
This is only partially true for \SC\ ---
due to the exact same reason, we easily find in the graph a cut whose
edge capacity exactly matches the corresponding expression in the LP.
However, the demands are arbitrary (and in particular do not have
a bounded-treewidth structure),
and analyzing them requires considerably more work.

\pparagraph{Hardness and integrality gaps for sparsest-cut.}

As mentioned earlier, \SC is known to be NP-hard~\cite{MS86},
and we further show in \ifprocs the full version \else Appendix~\ref{app:hardness} \fi
that it is even NP-hard on graphs of pathwidth $2$.
Two results~\cite{KV05,CKKRS06} independently proved that
under Khot's unique games conjecture~\cite{Khot02},
the \SC problem is NP-hard to approximate within any constant factor.
However, the graphs produced by the reductions in these two results have large treewidth.

The standard flow LP relaxation for \SC was shown in~\cite{LR99} to have
integrality gap $\Omega(\log n)$ in expander graphs, even for uniform demands.
Its standard strengthening to an SDP relaxation
(the SDP used by the known approximation algorithms of~\cite{ARV09,ALN08})
was shown in \cite{KV05,KR09,DKSV05} to have integrality gap
$\Omega(\log\log n)$, even for uniform demands.
For the case of general demands, a stronger bound $(\log n)^{\Omega(1)}$
was recently shown in~\cite{CKN09}.
% The LP relaxations derived using $k$ rounds of the Sherali-Adams procedure
% were shown in~\cite{CMM09} to have integrality gap
% $\tilde\Omega(\frac{\log n}{k})$, even for uniform demands.
% The relaxations derived similarly but when starting from the standard SDP,
% were shown in~\cite{RS09} to have a gap of $(\log\log n)^{\Omega(1)}$
% (for nontrivial $k$), again even for uniform demands.
Some of these results were extended in \cite{CMM09,RS09}
to certain hierarchies and a nontrivial number of rounds,
even for uniform demands.
Again, the graphs used in these results have large treewidth.

Integrality gaps for graphs of treewidth $r$
(or excluding a fixed minor of size $r$) follow from the above
in the obvious way of replacing $n$ with $r$ (or so),
for instance, the standard flow LP has integrality gap $\Omega(\log r)$.
However, no stronger gaps are known for these families;
in particular, it is possible that the integrality gap approaches $1$
with sufficiently many rounds (depending on $r$, but not on $n$).

\subsection{Discussion and further questions}
\label{sec:open}

We show that for the \SC problem,
the Sherali-Adams (SA) LP hierarchy can yield algorithms with
better approximation ratio than previously known.
Moreover, our analysis exhibits a strong (but rather involved)
connection between the input graph's treewidth and the SA hierarchy level.
Several interesting questions arise immediately:
\begin{enumerate} \compactify
\item Can this approach be generalized to excluded-minor graphs?
\item Can the approximation factor be improved to an absolute constant (independent of the treewidth)?
\end{enumerate}
A particularly intriguing and more fundamental question is
whether this hierarchy (or a related one, or for a different input family)
is strictly stronger than the standard LP (or SDP) relaxation. One possibility
is that our relaxation can actually yield an absolute constant factor approximation (as in Question~2).
Such an approximation factor is shown in \cite{CMM09}
to require at least $\Omega(\log r)$ rounds of Sherali-Adams,
and we would conclude that hierarchies yield strict improvement ---
higher (yet constant) levels of the Sherali-Adams hierarchy
do give improved approximation factors,
for an increasing sequence of graph families.
We note, however, that this would require a different rounding algorithm (see \ifprocs Remark~\ref{rem:LB}\else Remark~\ref{rem:LB} and Section~\ref{sec:lb}\fi).
Another possibility is that the GNRS conjecture does not hold even for bounded treewidth graphs, in which case the integrality gap of the standard LP exhibits a dependence on $n$, while, as we prove here, the stronger LP does not.

\iffalse
\begin{verbatim}
Eden said: one way to obtain
results similar to ours is to resolve the conjecture, which seems more
difficult (and in particular, is tied to embeddings, etc). We bypass
the need to resolve the conjecture by considering a stronger LP, but
it also raises interesting questions regarding whether this is an
instance where SA helps or doesn't help (where the answer depends both
on the conjecture, and the possibility of improving our rounding
algorithm).
\end{verbatim}
\fi

%%% Local Variables:
%%% mode: latex
%%% TeX-master: "BTWSC"
%%% End:

%% file: overview.tex
Relaxations arising from the Sherali-Adams (SA) hierarchy, and lift-and-project techniques in general, are known to give LP (or SDP) solutions which satisfy the following property: for every subset of variables of bounded size (bounded by the level in the hierarchy used), the LP/SDP solution restricted to these variables is a convex combination of valid $\{0,1\}$ assignments. Such a convex combination can naturally be viewed as a distribution on local assignments. In our case, for example, in an induced subgraph on $r+1$ vertices $S$, an $(r+1)$-level relaxation gives a local distribution on assignments $f:S\rightarrow\{0,1\}$ such that for every edge $(i,j)$ within $S$, the probability that $f(i)\neq f(j)$ is exactly the contribution of edge $(i,j)$ to the objective function (which we also call the \emph{LP-distance} of this pair). Our algorithm makes explicit use of this property, which is very useful for treewidth $r$ graphs.

Given an $(r+3)$-level Sherali-Adams relaxation, for every demand pair there is some distribution which (within every bag) matches the local distributions suggested by the LP, and also \emph{cuts/separates} this demand pair (i.e.\ assigns different values to its endpoints) with the correct probability (the LP distance). Unfortunately, there might not be any single distribution which is consistent with all demand pairs, so instead our algorithm assigns $\{0,1\}$ values at random to
the vertices of the graph $G$ in a stochastic process which
matches the local distributions suggested by the LP solution (per bag), but is oblivious to the structure of the demands $D$.

\iffalse
At a high level, our algorithm assigns  Since every edge is
contained in some bag, this rounding scheme cuts (assigns different
values to the endpoints of) exactly the number of edges (in
expectation) as the corresponding LP value. However, the algorithm's
assignments do not respect the distribution for pairs across different
bags.  In particular, this happens for demand pairs, which might not reside in the same bag, and may even be quite far from each other in the tree decomposition.
\fi

\pparagraph{Intuition.}
To achieve a good approximation ratio, it suffices to ensure that every demand pair is cut with probability not much smaller than the its LP distance. To achieve this, the algorithm fixes an arbitrary bag as the root, and traverses the tree decomposition one bag at a time, from the root towards the leaves, and samples the assignment to currently unassigned vertices in the current bag. This assignment is sampled in a way that ignores all previous assignments to vertices outside the current bag, but achieves the correct distribution on assignments to the current bag. Essentially, the algorithm finds locally correct distributions while maximizing the entropy of the overall distribution. Intuitively, this should only ``distort" the distribution suggested by the LP (for a given demand pair) only by introducing noise, which (if the noise is truly unstructured) mixes the correct global distribution with a completely random one in which every two vertices are separated with probability $\frac12$. In this case, the probability of separating any demand pair would decrease by at most a factor $2$. Unfortunately, we are not able to translate this intuition into a formal proof (and on some level, it is not accurate -- see Remark~\ref{rem:LB}). Thus we are forced to adopt a different strategy in analyzing the performance of the rounding algorithm. Let us see one illustrative special case.

\pparagraph{Example: Simple Paths.} Consider, for concreteness, the case of a single simple path $v_1,v_2,\ldots,v_n$. For every edge in the path $(v_{i-1},v_i)$, the LP suggests cutting it (assigning different values) with some probability $p_i$. Our algorithm will perform the following Markov process: pick some assignment $f(v_1)\in\{0,1\}$ at random according to the LP, and then, at step $i$ (for $i=2,\ldots,n$) look only at the assignment $f(v_{i-1})$ and let $f(v_i)=1-f(v_{i-1})$ with probability $p_i$, and $f(v_i)=f(v_{i-1})$ otherwise. Each edge has now been cut with exactly the probability corresponding to its LP distance. However, for $(v_1,v_n)$, which could be a demand pair, the LP distance between them might be much greater than the probability $q_n=\prob[f(v_1)\neq f(v_n)]$. Let us see that the LP distance can only be a constant factor more.% than the above probability.

First, if the above probability satisfies $q_n\geq\frac13$, then clearly we are done, as all LP distances will be at most 1. Thus we may assume that $q_n\leq\frac13$. Let us examine what happens at a single step. Suppose the algorithm has separated $v_1$ from $v_{i-1}$ with some probability $q_{i-1}\leq\frac13$ (assuming that all $q_i\leq\frac13$ is a somewhat stronger assumption than $q_n\leq\frac13$, but a more careful analysis shows it is also valid). After the current step (flipping sides with probability $p_i$), the probability that $v_i$ is separated from $v_1$ is exactly $(1-q_{i-1})p_i+q_{i-1}(1-p_i)$. This is an increase over the previous value $q_{i-1}$ of at least $$[(1-q_{i-1})p_i+q_{i-1}(1-p_i)]-q_{i-1}=(1-2q_{i-1})p_i\geq p_i/3.$$ However, the LP distance from $v_1$ can increase by at most $p_i$ (by triangle inequality). Thus, we can show inductively that we never lose more than a factor~3.

In general, our analysis will consider paths of bags of size $r+1$.
Even though we can still express the distribution on assignments
chosen by the rounding algorithm as a Markov process (where the
possible states at every step will be assignments to some set of at
most $r$ vertices), it will be less straightforward to relate the LP
values to this process. It turns out that we can get a handle on the
LP distances by modeling the Markov process as a layered digraph $H$
with edges capacities representing the transitions (this is only in the analysis,
or in the derandomization of our algorithm). In this case the
LP distance we wish to bound becomes the value of a certain
$(s,t)$-flow in $H$. We then bound the flow-value from above by
finding a small cut in $H$. Constructing and bounding the capacity of
such a cut in $H$ constitutes the technical core of this work.

%% file: algo.tex
\subsection{An LP relaxation using the Sherali Adams hierarchy}

Let us start with an informal overview of the Sherali-Adams (SA)
hierarchy. In an LP relaxation for a 0--1 program, the linear variables $\{y_i\mid
i\in[n]\}$ represent linear relaxations of integer variables
$x_i\in\{0,1\}$. We can extend such a relaxation to include variables
$\{y_I\}$ for larger subsets $I\subseteq[n]$ (usually, up to some
bounded cardinality). These should be interpreted as representing the
products $\prod_{i\in I}x_i$ in the intended (integer) solution. Now,
for any pair of sets $I,J\subseteq[n]$, we will denote by $y_{I,J}$ the linear relaxation for the polynomial $\prod_{i\in I}(1-x_i)\prod_{j\in J}x_j$. These can be derived from the variables $y_I$ by the inclusion-exclusion principle. That is, we define $$y_{I,J}=\textstyle\sum_{I'\subseteq I}(-1)^{|I'|}y_{I'\cup J}.$$ The constraints defined by the polytope $\SA_t(n)$, that is, level $t$ of the Sherali-Adams hierarchy starting from the trivial $n$-dimensional LP, are simply the inclusion-exclusion constraints:
\begin{eqnarray}\label{LP:inclusion-exclusion} \forall I,J\subseteq[n]\text{ s.t.\ }|I\cup J|\leq t\::\:y_{I,J}\geq 0%\\
%y_\emptyset=1
\end{eqnarray}
For every solution other than the trivial (all-zero) solution, we can define a normalized solution $\{\tilde{y}_I\}$ as follows:
$$\tilde{y}_I=y_I/y_\emptyset,$$ and the normalized derived variables $\tilde{y}_{I,J}$ can be similarly defined.

As is well-known, in a non-trivial level $t$ Sherali-Adams solution, for every set of (at most) $t$ vertices, constraints~\eqref{LP:inclusion-exclusion} imply a distribution on $\{0,1\}$ assignments to these vertices matching the LP values:

\begin{lemma}\label{lem:SA} Let $\{y_I\}$ be a non-zero vector in the
	polytope $\SA_t(n)$. Then for every set $L\subseteq[n]$ of
	cardinality $|L|\leq t$, there is a distribution $\mu_L$ on
	assignments $f:L\rightarrow\{0,1\}$ such that for all %
	%subsets
	$I,J\subseteq L$, %we have
	$$\prob_{\mu_L}\left[(\forall i\in I:f(i)=0)\;\wedge\;(\forall
	j\in J:f(j)=1)\right]=\tilde{y}_{I,J}.$$
\end{lemma}

In a Sparsest Cut relaxation, we are interested in the event in which a pair of vertices is cut (i.e.\ assigned different values). This is captured by the following linear variable: $$y_{i\neq j}=y_{\{i\},\{j\}}+y_{\{j\},\{i\}}.$$ We can now define our relaxation for Sparsest Cut, $\SCLP_r(G)$: % to be the following feasibility LP for the condition that the Sparsest Cut in $G$ has sparsity at most $\alpha$:
\begin{align}
%\min \sum_{(i,j)\in E}\capp(i,j)(y_i+y_j-2y_{ij}) & \leq & \alpha\sum_{(i,j)\in D}\dem(i,j)(y_i+y_j-2y_{ij}) \label{LP:sparsity}\\
\min\quad &\sum_{(i,j)\in E}\capp(i,j)y_{i\neq j} & \label{LP:obj-fun}\\
% \sum_{i,j\in V}y_{ij} & \geq & n-1\label{LP:nontrivial}\\
\text{s.t.}\quad &\sum_{i,j\in D}\dem(i,j)y_{i\neq j}=1&\label{LP:normalize}\\
 & \{y_I\}  \in \SA_{r+3}(n)&\label{LP:SA}\\
 & y_{I,J}= y_{J,I}&\forall I,J\text{ s.t. }|I\cup J|\leq{r+3}\label{LP:symmetry}
\end{align}

%Note that the term $y_i+y_j-2y_{ij}$ in constraint~\eqref{LP:sparsity} is exactly the probability that a distribution as in Lemma~\ref{lem:SA} assigns different values to vertices $i$ and $j$.
Note that constraint~\eqref{LP:normalize} is simply a normalization ensuring that the objective function is really a relaxation for the ratio of the two sums. % that the solution is not trivial (i.e.\ not all vertices are on the same side of the cut).
Also note that constraint~\eqref{LP:symmetry}, which ensures that the LP solution is fully symmetric, does not strengthen the LP, in the following sense: For any solution $\{y_I'\}$ to the above LP without constraint~\eqref{LP:symmetry}, a new solution to the symmetric LP (with the same value in the objective function) can be achieved by taking $y_I=(y'_I+y'_{I,\emptyset})/2$ without violating any of the other constraints. In particular, for every vertex $i\in V$ this gives $\tilde y_i=1-\tilde y_i=\frac12$. While our results hold true without imposing this constraint, we will retain it as it simplifies our analysis.

%\begin{remark} In a true Sherali-Adams relaxation, constraints~\eqref{LP:sparsity} and~\eqref{LP:nontrivial} would be ``lifted'' as well. We do not add such lifted constraints, as they would not help our analysis.
%\end{remark}
\begin{remark}\label{rem:LP-size}
The size of this LP (and the time needed to solve it) is $n^{O(r)}$. Specifically for bounded-treewidth graphs, we could also formulate a much smaller LP, where constraint~\eqref{LP:SA} would be replaced with the condition $\{y_I\mid I\subseteq B\cup\{i,j\}\}\in \SA_{r+3}(r+3)$ for every bag $B$ and demand pair $(i,j)\in D$. This would reduce the size of the LP to (and time needed to solve it) to at most $\poly(2^rn)$, and our rounding algorithm and analysis would still hold.
\end{remark}

\subsection{Rounding the LP}

Before we present the rounding algorithm, let us introduce some notation which will be useful in describing the algorithm. This notation will allow us to easily go back-and-forth between the LP solution and the local distributions on assignments described in Lemma~\ref{lem:SA}. For ease of notation, whenever two functions $f_1,f_2$ have disjoint domains, we will denote by $f_1\cup f_2$ the unique function from the union of the domains which is an extension of both $f_1$ and $f_2$.

\begin{itemize}
\item For every set of vectors $\{y_I\}$ and subset $L\subseteq[n]$ as
	in Lemma~\ref{lem:SA}, we will denote by
	$\distr^{\{y_I\}}_L$ the distribution on random
	assignments to $L$ guaranteed by the lemma.  We will omit the
	superscript $\{y_I\}$, and simply write $\distr_L$, when it is clear from the context.
\item Conversely, for any fixed assignment $f':L\rightarrow\{0,1\}$, we will write $\tilde y_{f'}=\tilde y_{L_0,L_1}$, where $L_b=\{i\in L\mid f'(i)=b\}$ for $b=0,1$. Thus, for a random assignment $f:L\rightarrow\{0,1\}$ distributed according to $\distr_L$, we have $\prob[f=f']=\tilde y_{f'}$.
\item For any nonempty subset $L'\subseteq L$, and a given assignment $f_0:L\setminus L'\rightarrow\{0,1\}$ in the support of $\distr_{L\setminus L'}$, we will denote by $\distr_{L',f_0}$ the distribution on random assignments % sampled according the distribution
    $f\sim\distr_L$ conditioned on the partial assignment $f_0$. Formally, a random assignment
    $f':L'\rightarrow \{0,1\}$, distributed according to $\distr_{L',f_0}$ satisfies $\prob_{f'}[f'=f_1]=\tilde y_{f_0\cup f_1}/\tilde y_{f_0}$ for every choice of $f_1:L'\rightarrow\{0,1\}$.
%
%    for any fixed assignment $f':L'\rightarrow\{0,1\}$, a random assignment $f:L'\rightarrow \{0,1\}$ distributed according to $\distr_{L',f_0}$ satisfies $\prob[f=f']=\tilde y_{f_0\cup f'}/\tilde y_{f_0}.$
\end{itemize}

Let $G$ be an graph with treewidth $r$ for some integer $r>0$, and let $(\mathcal B,T)$ be the corresponding tree decomposition. %For some positive real $\alpha>0$, l
Let $\{y_I\}$ be a vector satisfying $\SCLP_r(G)$. We now present the rounding algorithm:

\begin{center}\fbox{
\ifprocs\begin{minipage}{11.5 cm}\else\begin{minipage}{15 cm}\fi
Algorithm $\SCRound(G,(\mathcal B,T),\{y_I\})$\quad
[Constructs a random assignment \ifprocs$f$\else$f:V\rightarrow\{0,1\}$\fi]
\begin{enumerate}
\ifprocs
    \item Pick an arbitrary $B_0\in\mathcal B$ as the root of $\mathcal T$, and sample $f|_{B_0}\sim\distr_{B_0}$.
\else
    \item Pick an arbitrary bag $B_0\in\mathcal B$ as the root of $\mathcal T$, and sample $f|_{B_0}$ according to $\distr_{B_0}$.
\fi
    \item Traverse the rest of the tree $T$ in any order from the root towards the leaves. For each bag $B$ traversed, do the following:
    \begin{enumerate}
        \item\label{step:separate} Let $B^+$ be the set of vertices in $B$ for which $f$ is already defined, and let $B^-=B\setminus B^+$. Let $f_0$ be the \ifprocs existing \else corresponding \fi assignment $f_0=f|_{B^+}$.
        \item\label{step:conditional-assign} If $B^-$ is non-empty, %for $b=0,1$, let $B^+_b=\{i\in B^+\mid f(i)=b\}$, and choose values for $f(\cdot)$ on $B^-$ at random such that for all disjoint subsets $I,J\subseteq B^-$, $$\prob[\forall i\in I:f(i)=0\wedge\forall j\in J:f(j)=1]=\frac{y_{I\cup B^+_0,J\cup B^+_0}}{y_{I,J}}.$$
sample $f|_{B^-}$ at random according to $\distr_{B^-,f_0}$.
    \end{enumerate}
\end{enumerate}
\end{minipage}
}%\fbox
\end{center}

Let us first see that every edge $(i,j)\in E$ is cut with probability exactly $\tilde y_{i\neq j}$. Since every edge is contained in at least one bag, it suffices to show that within every bag $B$, the assignment $f|_B$ is distributed according to $\distr_B$. This is shown by the following \ifprocs lemma, whose straightforward proof appears in the full version. \else straightforward lemma. \fi

\begin{lemma}\label{lem:bags} For every bag $B$, the assignment $f|_B$ produced by \ifprocs running \else\fi algorithm $\SCRound(G,(\mathcal B,T),\{y_I\})$ is distributed according to $\distr_B$.
\end{lemma}
\ifprocs\else
\begin{proof} We show this by induction. For $B_0$ this holds as the assignment $f|_{B_0}$ is explicitly sampled according to this distribution.

Now, let $B$ be a new bag traversed and $B^+$ and $B^-$ as in Step~\ref{step:separate}. By the definition of a tree decomposition, and since the tree traversal maintains a single connected component, $B^+$ must be fully contained in some bag $B'$ which has already been traversed. Thus, by the inductive hypothesis, $f|_{B'}$ is distributed according to $\distr_{B'}$, and in particular, $f|_{B^+}$ is distributed according to $\distr_{B^+}$. Note that this is also the marginal distribution of assignments to $B^+$ according to $\distr_{B}$. Thus, the assignment to $f|_{B^+}$ must lie in the support of $\distr_{B^+}$ (this shows that Step~\ref{step:conditional-assign} is well defined), and for every such fixed assignment $f_0$, and every fixed assignment $f_1:B^-\rightarrow\{0,1\}$, we have
\begin{align*}
\prob[f|_B=f_0\cup f_1]&=\prob[f|_{B^+}=f_0]\cdot\prob[f|_{B^-}=f_1\mid f|_{B^+}=f_0]\\
&=\prob_{f'\sim\distr_{B}}[f'|_{B^+}=f_0]\cdot\prob_{f'\sim\distr_{B^-,f_0}}[f'|_{B^-}=f_1]\\
&=\prob_{f'\sim\distr_{B}}[f'|_{B^+}=f_0]\cdot\prob_{f'\sim\distr_{B}}[f'|_{B^-}=f_1\mid f'|_{B^+}=f_0]\\
&=\prob_{f'\sim\distr_{B}}[f'|_B=f_0\cup f_1].
\end{align*}
\end{proof}
\fi

This lemma shows that the expected value of the cut is %exactly the expression
 $\sum_{(i,j)\in E}\capp(i,j)\tilde y_{i\neq j}$, which is exactly the value of the objective function~\eqref{LP:obj-fun} scaled by $1/y_{\emptyset}$. In particular, for a host of other problems where the objective function and constraints depend only on the edges (e.g.\ Minimum Vertex Cover, Chromatic Number), this type of LP relaxation (normalized by setting $y_\emptyset=1$), along with the above rounding, always produces an optimal solution for bounded-treewidth graphs. Thus, in some sense, we consider this to be a ``natural'' rounding algorithm.

Before we analyze the expected value of the cut demands (or specifically, the probability that each demand is cut), let us show that the order in which the tree $T$ is traversed has no effect on the distribution of cuts produced (it will suffice to show a slightly weaker claim -- that the joint distribution of cuts in any two bags is not affected). This is shown in the following \ifprocs lemma, whose proof appears in the appendix. \else lemma. \fi

\begin{lemma}\label{lem:order-invariant} Let $B_1,B_2\in\mathcal B$ be two arbitrary bags. Then the distribution on assignments $f|_{B_1\cup B_2}$ is invariant under any connected traversal of $T$.
\end{lemma}
\ifprocs\else
\begin{proof} Since the set of traversed bags is always a connected component, it suffices to consider only the tree-path connecting $B_1$ and $B_2$. Let us proceed by induction on the length of the path. If $B_1$ and $B_2$ are adjacent bags, then regardless of the order in which they are traversed, for any fixed assignment $f_0:B_1\cap B_2\rightarrow\{0,1\}$, if $f|_{B_1\cap B_2}=f_0$, then $f_{B_1\setminus B_2}$ and $f_{B_2\setminus B_1}$ are distributed according to $\distr_{B_1\setminus B_2,f_0}$ and $\distr_{B_2\setminus B_1,f_0}$, respectively. Moreover, these two assignments are independent (after conditioning on $f|_{B_1\cap B_2}=f_0$) regardless of the order of traversal. Thus, in either ordering, the same distribution on cuts can be achieved by first sampling $f|_{B_1\cap B_2}$ according to $\distr_{B_1\cap B_2}$, and then sampling $f_{B_1\setminus B_2}$ and $f_{B_2\setminus B_1}$ independently according to the above distributions, where $f_0=f|_{B_1\cap B_2}$.

Now suppose that bags $B_1$ and $B_2$ are not adjacent, and let $B'_1$ the bad adjacent to $B_1$ on the path to $B_2$. By the inductive hypothesis, the distribution on $f|_{B'_1\cup B_2}$ is invariant under the order in which the path between $B'_1$ and $B_2$ is traversed. In particular, this is true for the distribution on $f|_{(B_1\cap B'_1)\cup (B_2\setminus B_1)}$, call it $\mathcal{F'}$. Thus, arguing as above, we see that any ordering results in the distribution on $f_{B_1\cup B_2}$ obtained by first sampling $f|_{B_1\cap B'_1}$ according to $\distr_{B_1\cap B'_1}$, and then sampling $f_{B_1\setminus B'_1}$ according to $\distr_{B_1\setminus B'_1,f|_{B_1\cap B'_1}}$ and then independently sampling $f_{B_2\setminus B_1}$ according to the distribution $\mathcal{F'}$ conditioned on the value of $f|_{B_1\cap B'_1}$.
\end{proof}
\fi

%%% Local Variables:
%%% mode: latex
%%% TeX-master: "BTWSC"
%%% End:

%% file: flow.tex
\ifprocs We show the following lemma, whose proof appears in the full version, \else In this section and the next two, we shall show the following lemma, \fi which together with Lemma~\ref{lem:bags} implies Theorem~\ref{thm:main} (see Remark~\ref{rem:derandomize}).
\begin{lemma}\label{lem:demands} For every integer $r>0$ there exists a constant $c_r>0$ such that for any treewidth-$r$ graph $G$ with tree decomposition $(\mathcal B,T)$, and vectors $\{y_I\}$ satisfying SC$_r(G)$, algorithm SC-Round$(G,(\mathcal B,T),\{y_I\})$ outputs a random %assignment
$f:V\rightarrow\{0,1\}$ s.t.\ for every %two vertices
$i,j\in V$,% we have
\begin{equation}\label{eq:pairs-main}\prob[f(i)\neq f(j)]\geq c_r\tilde y_{i\neq j}.
\end{equation}
\end{lemma}

\begin{remark}\label{rem:LB} The constant $c_r$ arising in our analysis is quite small (roughly $2^{-r2^r}$). While we believe this can be improved, we cannot eliminate the dependence on $r$, as a lower bound on the performance of our rounding algorithm \ifprocs (which appears in the full version) \else (see Section~\ref{sec:lb}) \fi shows that $c_r$ cannot be more than $2^{-r/2}$.
\end{remark}

\begin{remark}\label{rem:derandomize} In fact, Lemmas~\ref{lem:bags} and~\ref{lem:demands} taken together show the following: Given any solution to SC$_r(G)$ with objective function value $\alpha>0$, algorithm SC-Round produces a random assignment $f$ satisfying $$\expec\left[\sum_{(i,j)\in E}\capp(i,j)\left|f(i)-f(j)\right|-{\textstyle\frac{\alpha}{c_r}}\sum_{(i,j)\in D}\dem(i,j)\left|f(i)-f(j)\right|\right]\leq 0.$$ This means the algorithm produces a $1/c_r$-approximation with positive probability, but does not immediately imply a lower bound on that probability. Fortunately, following the analysis in this section, the algorithm can be derandomized by the method of conditional expectations, since, at each step, finding the probability of separating each demand pair reduces to calculating the probability of reaching a certain state at a certain phase in some Markov process, which simply involves multiplying $O(n)$ transition matrices of size at most $2^r\times 2^r$ (in fact, these can be consolidated so that every step of the algorithm involves a total of $O(n|T|)$ small matrix multiplications for all demands combined, where $T$ is the set of vertices participating in demand pairs).
\end{remark}

For vertices $i,j\in V$ belonging to (at least) one common bag,
Lemma~\ref{lem:bags} implies equality in~\eqref{eq:pairs-main} for
$c_r=1$. For $i,j\in V$ which do not lie in the same bag, consider the
path of bags $B_1,\ldots,B_N$ in tree $T$ from the (connected)
component of bags containing $i$ to the component of bags containing $j$. By Lemma~\ref{lem:order-invariant}, we may assume that the algorithm traverses the %bags in this
 path in order from $B_1$ to $B_N$.

To understand the event that vertices $i$ and $j$ are separated, it suffices to consider the following incomplete (but consistent) description of the stochastic process involved: Let $S_0=\{i\}$ and $S_N=\{j\}$, and let $S_l=B_l\cap B_{l+1}$ for $l=1,\ldots,N-1$. The algorithm assigns $f(i)$ a value in $\{0,1\}$ uniformly at random, and then for $l=1,\ldots,N$, samples $f|_{S_l}$ from the distribution $\distr_{S_l,f|_{S_{l-1}}}$ (we extend the definition of $\distr_{S,f'}$ in the natural way to include the case where $S$ may intersect the domain of $f'$).

This is a Markov process, and can be viewed as a Markov flow graph. That is, a layered graph, where each layer consists of nodes representing the different states (in this case, assignments to $S_l$), with exactly one unit of flow going from the first to the last layer, with all edges having flow at full capacity. Since all edges in the flow graph represent pairs of assignments within the same bag, Lemma~\ref{lem:bags} implies that the capacity of an edge (transition) $(f_1,f_2)$ is exactly $\tilde y_{f_1\cup f_2}$, and the amount of flow going through each node $f_0$ is $\tilde y_{f_0}$. \ifprocs\else For the sake of clarity, we will refer to the elements of $V$ as \emph{vertices} and to the states in the flow graph as \emph{nodes}.\fi

We now would like to analyze the contribution of a demand pair to the LP. % (we may assume all pairs are demand pairs, possibly with zero weight).
 By constraint~\eqref{LP:symmetry}, this contribution (up to a factor $\dem(i,j)$) is $\tilde{y}_{i\neq j}=2\tilde y_{\{i\},\{j\}}=2\tilde y_{f^*}$, where $f^*:\{i,j\}\rightarrow\{0,1\}$ is the function assigning 0 to $i$ and 1 to $j$. Now consider a layer graph as above where each edge $(f_1,f_2)$ has flow $\tilde y_{f^*\cup f_1\cup f_2}$. To see that this is indeed a flow, note that two consecutive layers along with $i$ and $j$ only involve at most $r+3$ vertices in $G$, and so by Lemma~\ref{lem:SA} for any $l>0$ and function $f_2:S_l\rightarrow\{0,1\}$ the incoming flow at $f_2$ must be $\displaystyle\sum_{f_1\in S_{l-1}}\tilde y_{f^*\cup f_1\cup f_2}=\tilde y_{f^*\cup f_2}$, and so is the outgoing flow. The total flow in this graph is exactly $\tilde y_{f^*}$ (half the LP contribution $\tilde y_{i\neq j}$). Moreover, for each such edge (transition) we also have $\tilde y_{f^*\cup f_1\cup f_2}\leq \tilde y_{f_1\cup f_2}$. Hence, the flow with values $\{\tilde y_{f^*\cup f_1\cup f_2}\}$ is a legal flow respecting the capacities $\{\tilde y_{f_1\cup f_2}\}$ in the Markov flow graph which represents the rounding algorithm.

Thus it suffices to show the \ifprocs following theorem (proved in the full version): \else following: \fi
\begin{theorem}\label{thm:flow-main} For every integer $k>1$, there is
	a constant $C=C(k)>0$ such that for any symmetric Markov flow
	graph $G=(L_0,\ldots,L_N,E)$  representing a Markov process
	$X_0,\ldots,X_N$ with sources $L_0=\{s_0,s_1\}$ and sinks
	$L_N=\{t_0,t_1\}$ and at most $k$ nodes per layer, the total
	amount of capacity-respecting flow in $G$ from $s_0$ to $t_1$ can be at most $C\cdot\prob[X_0=s_0\wedge X_N=t_1]$.
\end{theorem}

\iffalse
at in a symmetric Markov flow graph as above, with two sources $s_0,s_1$ and two sinks $t_0,t_1$, the amount of capacity respecting flow from $s_0$ to $t_1$ can be at most a constant $C$ factor more than the amount of flow from $s_0$ to $t_1$ in the original (Markov) flow, for some constant $C=C(r)>0$. \fi

Applying this theorem to the Markov flow graph described above with $k=2^r$ immediately implies Lemma~\ref{lem:demands}. As usual, to bound the amount of flow in a graph from above, it suffices to find a suitable \ifprocs cut. See the full version for details. \else cut, which is what we will do in the following section.\fi

%% file: cut.tex
In this section we prove Theorem~\ref{thm:flow-main} for $k=4$ (the proof for the general case appears in Section~\ref{sec:mainproof}). For Markov flow graph $G=(L_0,\ldots,L_N,E)$ and corresponding Markov process $X_0,\ldots,X_N$ as in the theorem, for any integers $0\leq l_1\leq l_2\leq N$, and any vertices $u\in L_{l_1},v\in L_{l_2}$ we will let $p(u)=\prob[X_{l_1}=u]$ be the probability of reaching $u$, and similarly, we define $p(u,v)=\prob[X_{l_1}=u\wedge X_{l_2}=v]$. In particular, when $l_2=l_1+1$ and $\overrightarrow{(u,v)}$ is an edge (transition) then $p(u,v)$ is also the capacity of this edge. Note that, by the symmetry of $G$, we have $p(s_0)=p(s_1)=\frac12$.

\subsection{A potential function for Markov flow graphs}

\iffalse
In this section$G=(L_0,\ldots,L_N,E)$ is a symmetric Markov flow graph, representing a Markov process $X_0,\ldots,X_N$ with $L_0=\{s_0,s_1\}$ and $L_N=\{t_0,t_1\}$, and $L_l=\{v^l_i\}$ for $l=1,\ldots,t-1$. Each layer has at most $k$ vertices. Each vertex $v$ has probability $p(v)$ of being reached in the Markov flow starting from $p(s_0)=p(s_1)=\frac12$, and each transition $\overrightarrow{(u,v)}$ has probability $p(u,v)$ of being traversed (i.e.\ capacity). We extend this notation and write, for any vertices $u\in L_{l_1}, v\in L_{l_2}$, where $l_1<l_2$, $p(u,v)=\prob[X_{l_1}=u\wedge X_{l_2}=v]$.
\fi

We will define  a potential function on the layers of $G$, which will allow us to rephrase Theorem~\ref{thm:main} in more convenient terms. First, for any every layer $l$ and vertex $v\in L_l$, let us define $$A(v)=\prob[X_0=s_0\mid X_l=v]-\textstyle\frac12.$$
%\iffalse
This function satisfies the following stochastic property:% (the proof of which is immediate and appears in the appendix):

\begin{lemma}\label{lem:A-stochastic} For and $0<l_1<l_2$ and $v\in L_{l_2}$ we have $$A(v)=\frac{\sum_{u\in L_{l_1}}p(u,v)A(u)}{\sum_{u\in L_{l_1}}p(u,v)}.$$
\end{lemma}
\ifprocs\else
\begin{proof} By the Markov property, we have
\begin{align*}
A(v)+{\textstyle\frac12}=\prob[X_0=s_0\mid X_{l_2}=v]=p(s_0,v)/p(v)%\\
&={\textstyle\frac1{p(v)}}{\sum_{u\in L_{l_2}}(p(u,v)/p(u))p(s_0,u)}\\
&={\textstyle\frac1{p(v)}}{\sum_{u\in L_{l_2}}p(u,v)(A(u)+{\textstyle\frac12})}\\
&=\frac{\sum_{u\in L_{l_2}}p(u,v)(A(u)+{\textstyle\frac12})}{\sum_{u\in L_{l_2}}p(u,v)}.
\end{align*}
\end{proof}
\fi

Now, for every layer $l=0,\ldots, N$, let us define the following potential function:
$$\varphi(l)=\var[A(X_l)]={\textstyle\sum}_{v\in L_l}p(v)A(v)^2-\left({\textstyle\sum}_{v\in L_l}p(v)A(v)\right)^2.$$
The following lemma show that this potential function is monotone decreasing in $l$, and relates the decrease directly to the transitions in the Markov process:% (proof appears in the appendix):

\begin{lemma}\label{lem:potential} For all $0<l_1<l_2$, we have
\ifprocs $$\varphi(l_1)-\varphi(l_2)=\displaystyle\sum_{u\in L_{l_1},v\in L_{l_2}}p(u,v)(A(u)-A(v))^2.$$
\else $\varphi(l_1)-\varphi(l_2)=\displaystyle\sum_{u\in L_{l_1},v\in L_{l_2}}p(u,v)(A(u)-A(v))^2.$
\fi
\end{lemma}
\ifprocs\else
\begin{proof}
By Lemma~\ref{lem:A-stochastic}, we have $$\sum_{v\in L_{l_2}}p(v)A(v)=\sum_{v\in L_{l_2}}\sum_{u\in L_{l_1}}p(u,v)A(u)=\sum_{u\in L_{l_1}}(\sum_{v\in L_{l_2}}p(u,v))A(u)=\sum_{u\in L_{l_1}}p(u)A(u).$$
Therefore, we have
\begin{align*}
\varphi(l_1)-\varphi(l_2)&=\sum_{u\in L_{l_1}}p(u)A(u)^2-\sum_{v\in L_{l_2}}p(v)A(v)^2\\
&=\sum_{u%\in L_{l_1}
                   }\sum_{v%\in L_{l_2}
                                  }p(u,v)A(u)^2-\sum_{v%\in L_{l_2}
                                                                         }p(v)A(v)^2&\text{since }p(u)=\sum_vp(u,v)\\
&=\sum_{u%\in L_{l_1}
                   }\sum_{v%\in L_{l_2}
                                  }p(u,v)\left(A(u)^2+A(v)^2\right)-2\sum_{v%\in L_{l_2}
                                                                          }p(v)A(v)^2&\text{since }p(v)=\sum_up(u,v)\\
&=\sum_{v%\in L_{l_2}
                   }\sum_{u%\in L_{l_1}
                                  }p(u,v)\left(A(u)^2+A(v)^2-2A(u)A(v))\right).&\text{by Lemma~\ref{lem:A-stochastic}}
\end{align*}
\end{proof}
\fi

Recall that we want to bound the possible flow from $s_0$ to $t_1$ by $O(p(s_0,t_1))$.  We may assume that  $p(s_0,t_1)<\frac14$ (i.e.\ $A(t_1)<0$), since otherwise the bound is trivial. Note that, by symmetry, we have $A(s_0)=-A(s_1)$ and $A(t_0)=-A(t_1)$. %Along with the fact that
Since we have only two sources and two sinks, this implies
 %We also need to mention that $\varphi(0)=\frac14$, and that, due to symmetry,
\begin{align*}\varphi(0)-\varphi(N)=A(s_0)^2-A(t_1)^2={\textstyle\frac14}-A(t_1)^2={\textstyle\frac14}-(2p(s_0,t_1)-{\textstyle\frac12})^2\leq2p(s_0,t_1).
\end{align*}

Therefore, to prove Theorem~\ref{thm:flow-main} it suffices to show
\begin{lemma}\label{lem:flow-main} For every $k>0$ there is some constant $C=C(k)$ such that for any $G$ as above, with $A(t_1)<0$, there is a cut in $G$ separating $s_0$ from $t_1$ of capacity at most $C\cdot(\varphi(0)-\varphi(N))$.
\end{lemma}

Symmetry implies that $k$ is even, and the case of $k=2$ is fairly trivial. We will first consider the simpler case of $k=4$, while the general case is shown in Section~\ref{sec:mainproof}.

\subsection{Treewidth 2}
Let us start with the case of $k=4$, or $r=2$. This shows some of the main ideas in the analysis for larger $k$, while still being relatively simple. It is also an non-trivial special case, as it covers series-parallel graphs. A more careful analysis would yield a smaller constant C (we did not optimize).%, we did not attempt to optimize for this case.

\begin{lemma}\label{lem:flow-k=4} Lemma~\ref{lem:flow-main} holds for $k=4$ and $C=100$.
\end{lemma}
\ifprocs\else
\begin{proof} We may assume that $\varphi(N)\geq\frac{49}{200}$. Otherwise, since the capacity of $s_0$ is $\frac12$, the cost of simply cutting the outgoing edges of $s_0$ is ${\textstyle\frac12}=C/200\leq C(\varphi(0)-\varphi(N)).$ Let us denote $A^*=A(t_0)=\sqrt{\varphi(N)}(\geq\frac7{10\sqrt2})$. Let us start by cutting all edges $(u,v)$ for which $|A(u)-A(v)|\geq\frac27A^*$. By Lemma~\ref{lem:potential}, the total capacity of these edges is at most
\begin{equation}\label{eq:k=4edges}
\begin{split}
\sum_{l=1}^{t}\sum_{\substack{u\in L_{l-1},v_{L_l}\\|A(u)-A(v)|\geq\frac27A^*}}p(u,v)
&\leq({\textstyle\frac{7}{2A^*}})^2\sum_{l=1}^{t}\sum_{\substack{u\in L_{l-1},v_{L_l}\\|A(u)-A(v)|\geq\frac27A^*}}p(u,v)(A(u)-A(v))^2\\
%&\leq({\textstyle\frac{7}{2A^*}})^2\sum_{l=1}^{t}\sum_{u\in L_{l-1},v_{L_l}}p(u,v)(A(u)-A(v))^2\\
&\leq({\textstyle\frac{7}{2A^*}})^2\sum_{l=1}^{t}(\varphi(l-1)-\varphi(l))=({\textstyle\frac{7}{2A^*}})^2(\varphi(0)-\varphi(N)).
\end{split}
\end{equation}

Let us examine the rest of the graph. By symmetry, and since $\varphi$ is monotone decreasing, every layer $l$ must contain some vertex $v_l$ such that $A(v_l)>A(t_0)=A^*$, and a corresponding vertex $\tilde{v_l}$ with $A(\tilde{v_l})=-A(v_l)$. Consider the inner two vertices $u_l,\tilde u_l$ (with $A(v_l)\geq A(u_l)=-A(\tilde u_l)\geq0$). Since there are no direct edges from $v_{l-1}$ to $\tilde v_l$ (the edge would be longer than $\frac27A^*$), the flow must travel along paths using the vertices $\{u_l,\tilde u_l\}_l$. Since we have cut long edges, the $A(\cdot)$ values of these paths must pass through the interval $[-\frac17A^*,\frac17A^*]$. Let $l_1$ be the first such interval for which there is flow from $s_0$ to $u_{l_1}$. By a similar argument, any flow from $s_0$ to $u_{l_1}$ must pass through the interval $[\frac37A^*,\frac57A^*]$ (before layer $l_1$). Let us take the last such layer, say $l_0$ (it can be checked that all flow to $u_{l_1}$ and $\tilde u_{l_1}$ must pass through $u_{l_0}$). To cut all flow to $u_{l_1},\tilde u_{l_1}$, it suffices to remove vertex $u_{l_0}$, or equivalently, to cut all outgoing edges from $u_{l_0}$. Note that for all vertices $w\in L_{l_1}$ we have $|A(w)-A(u_{l_0})|\geq\frac27A^*$, since $A(u_{l_0})\in [\frac37A^*,\frac57A^*]$, $A(v_{l_1})>A^*$, and $A(\tilde v_{l_1})\leq A(\tilde u_{l_1})\leq A(u_{l_1})\leq \frac17A^*$. Hence, by Lemma~\ref{lem:potential}, the cost of cutting vertex $u_{l_0}$ is at most
\begin{equation}
\begin{split}
p(u_{l_0})=\sum_{w\in L_{l_1}}p(u_{l_0},w)
&\leq({\textstyle\frac{7}{2A^*}})^2\sum_{w\in L_{l_1}}p(u_{l_0},w)(A(u_{l_0})-A(w))^2\\
%&\leq({\textstyle\frac{7}{2A^*}})^2\sum_{x\in L_{l_0},w\in L_{l_1}}p(x,w)(A(x)-A(w))^2\\
&\leq({\textstyle\frac{7}{2A^*}})^2(\varphi(l_0)-\varphi(l_1)).
\end{split}
\end{equation}

It is easy to see that any more flow from $s_0$ to $t_1$ must start at $v_{l_2}$ for some layer $l_2\geq l_1$, and so we can repeat the above argument, cutting vertices with $A(\cdot)$ value in $[\frac37A^*,\frac57]$, and paying $(\frac7{2A^*})^2(\varphi(l_i)-\varphi(l_{i+1}))$  each time for non-overlapping intervals $[l_0,l_1],[l_2,l_3],\ldots,[l_m,l_{m+1}]$, until we have severed all flow. Combining this with the cost incurred in~\eqref{eq:k=4edges}, we can bound the total capacity of edges cut by $2({\textstyle\frac7{2A^*}})^2(\varphi(0)-\varphi(N))%={\textstyle\frac{49}{2(A^*)^2}}(\varphi(0)-\varphi(N))
\leq100(\varphi(0)-\varphi(N)).$
\end{proof}
\fi

To summarize the above approach, our cutting technique follows a two phase process. First, we cut all ``long" edges, which helps us isolate individual paths in the flow. Then, we isolate portions of the graph where the individual paths have a large shift in $A(\cdot)$ value (e.g. move from the interval $[\frac37A^*,\frac57A^*]$ to the interval $[-\frac17A^*,\frac17A^*]$), and cut such paths by removing a single vertex, charging to the difference in potential along that portion of the graph.

There are a number of technical difficulties involved in extending this argument to work for larger $k$. First, we cannot isolate specific intervals through which flow must pass in an isolated path, as these depend on the $A(\cdot)$ values of other vertices in nearby layers. Secondly, even after cutting a path at some node, we are not guaranteed that there is no other path which routes flow around the node we cut. Rather than decompose the graph into isolated paths, we cut in several (roughly $k$) phases using a cut-and-cluster approach. Namely, after cutting, we ``cluster" together all vertices (or clusters from the previous phase) in a single layer that are close together in $A(\cdot)$ value, and ensure that the number of clusters per layer which can contain flow from $s_0$ to $t_1$ decreases after every phase.

Unfortunately, our threshold for clustering vertices increases by roughly a $k^2$ factor after every phase, thus ultimately incurring a loss which is exponential in $k$ (or doubly-exponential in $r$). %We believe that a more refined analysis may yield a bound which is only quasi-polynomial in $k$ (or $2^{\poly(r)}$), which is not far from our $\Omega(k)(=\Omega(2^r))$ lower-bound for our rounding algorithm.
We note that while this does not match our $\Omega(k)(=\Omega(2^r))$ lower-bound, the lower-bound at least shows that we can not expect to get any ``reasonable" dependence on $r$ (say, $O(\log r)$) with our rounding.

%% file: mainproof.tex
\section{Performance guarantee for general bounded treewidth}\label{sec:mainproof}

Let us begin with a simple lemma which was implicit in the analysis of Lemma~\ref{lem:flow-k=4}.

\begin{lemma}\label{lem:cut-charge} Let  $\{[l^0_j,l^1_j]\}$ be a sequence of non-overlapping intervals for integers $0\leq l^0_j<l^1_j\leq N$ and let $W_j\subseteq L_{l^0_j}$ be sets of nodes in layer $L_{l^0_j}$ such that for every node $w\in W_j$ and every node $x\in L_{l^1_j}$ we have $|A(w)-A(x)|\geq\rho$ for some $\rho>0$. Then the cost of removing all $w\in W_j$ for every $j$ (i.e.\ cutting all outgoing edges from $w$) is at most $\frac{1}{\rho^2}(\varphi(0)-\varphi(N))$.
\end{lemma}

As we will use this lemma repeatedly, for a set of edges (resp.\ nodes) $T$, we will call the value $p(T)/(\varphi(0)-\varphi(N))$ the \emph{relative cost} of $T$, where $p(T)$ is the total capacity of the edges (resp.\ nodes) in $T$. Thus our goal will be to find a cut of constant relative cost.

Let us introduce some terminology and notation:

\begin{definition} In the context of this section, the \emph{distance} between two nodes $u,v$ will always refer to the value $|A(u)-A(v)|$, which we will also call the $\emph{length}$ of $(u,v)$ when $(u,v)$ is an edge. For two non-overlapping clusters (defined below), we define the distance between them to be the minimum distance between two nodes, one in each cluster.
\end{definition}

\begin{definition} For any $\eps>0$, an $\eps$-\emph{cluster} is a set of nodes belonging to a single layer, such that when ordered by their respective $A(\cdot)$-values, every two consecutive nodes are at distance at most $\eps$ from each other. For any cluster $X$, we will denote $A^+(X)=\max_{v\in X}A(v)$ and $A^-(X)=\min_{v\in X}A(v)$, and we will refer to the value $A^+(X)-A^-(X)$ as the \emph{width} of $X$.
\end{definition}

Note that the width of any $\eps$-cluster is at most $(k-1)\eps$.

\begin{definition} In a Markov flow graph as above, with some edges already cut, we will call a cluster $X$ \emph{viable} if there is any capacity-respecting flow in the remaining graph from $s_0$ to $t_1$ which goes through at least one node of $X$. For any clustering of the graph, we will define the \emph{clustered capacity} to be the maximum number of viable clusters per layer, over all layers.
\end{definition}

Let us now prove Lemma~\ref{lem:flow-main}

\begin{proof}[Proof of Lemma~\ref{lem:flow-main}]
Let us assume that $A(t_1)<-\frac13$ (recall that we've assumed $A(t_1)\leq0$). Otherwise, simply cutting the outgoing edges of $s_0$ yields a cut of relative cost $\frac12/(\varphi(0)-\varphi(N))\leq\frac{18}5$.

As discussed earlier, we will proceed in $k$ phases. In each phase, we will reduce the clustered capacity of the graph. Each phase will consist of first cutting some clusters (i.e.\ removing all outgoing edges from the nodes in these clusters), and then increasing the size of certain other clusters (by increasing the threshold for clustering).

We begin by first cutting all edges of length at least $\eps_0$ (for some $\eps_0>0$ to be determined soon). By Lemma~\ref{lem:cut-charge}, the relative cost of this cut is at most $1/\eps_0^2$. At the end of each phase $j$, we will cluster the nodes with clustering threshold $\eps_j=(12k^2)^j\eps_0$, while we will require that $k\eps_{k-1}\leq\frac16$. Thus we set $\eps_0=1/(6k(12k^2)^{k-1})$. As we shall see, the relative cost of the cut at phase $j$ will be  at most $1/(k\eps_{j-1})^2$. Thus, the total relative cost of our cut will be at most $$\frac1{\eps_0^2}+\sum_{j>0}\frac1{k^2\eps_{j-1}^2}=\frac1{\eps_0^2}\left(1+\sum_{j>0}\frac1{144^{j-1}k^{4j-2}}\right)=O\left(\frac1{\eps_0^2}\right)=k^{O(k)}.$$

Before describing and analyzing the individual phases, let us note that at phase $j$ every cluster has width at most $(k-1)\eps_{j-1}$. Since we have cut all edges of length at least $\eps_0$, for two clusters $X_1,X_2$ in consecutive layers, there can be flow from $X_1$ to $X_2$ only if $A^-(X_2)-A^+(X_1)<\eps_0$ and $A^-(X_1)-A^+(X_2)<\eps_0$. In particular, when there is such flow, we have
\begin{equation}\label{eq:cluster-step}
\max\{|A^+(X_2)-A^+(X_1)|,|A^-(X_2)-A^-(X_1)|\}\leq (k-1)\eps_{j-1}+\eps_0\leq k\eps_{j-1}.
\end{equation}

We now proceed by induction on the clustered capacity. If the clustered capacity is 1, then there is a single ``path" of clusters from $s_0$ to $t_1$. Let $j$ be the current phase $(1\leq j\leq k)$. By our choice of $\eps_{j-1}$, and by \eqref{eq:cluster-step}, the value $A^+(X)$ of any ($\eps_{j-1}$-)cluster in the path can increase or decrease by at most $k\eps_{k-1}\leq\frac16$ at each step. Since this value starts at $A(s_0)=\frac12$, and ends at $A(t_1)<-\frac13$, at some point it must pass through the interval $[0,\frac16]$. Call this layer $l_1$, and the corresponding cluster $X_{l_1}$. Now $A^+(X_{l+1})\leq\frac16$, and since the width of this cluster is at most $\frac16$, we also have $A^-(X_{l+1})\geq-\frac16$. Thus all nodes in the cluster are at distance at least $\frac16$ from $t_0$ and $t_1$. Thus, by Lemma~\ref{lem:cut-charge}, the relative cost of cutting all flow along the path by removing $X_{l_1}$ is at most 36.

Now suppose the clustered capacity is $k'$ for some $1<k'\leq k$, and let $j$ denote the current phase. For any given layer, if the number of viable $\eps_{j-1}$-clusters is strictly less than $k'$, then we are done with that particular layer. Otherwise, if there are $k'$ viable clusters in a layer and any two of them are at distance at most $\eps_j$ from each other, then again we are done with that layer, since at the end of the phase the two clusters will be merged (possibly along with additional clusters) into a single $\eps_j$-cluster. Thus, we only need to reduce the number of viable clusters in layers which contain $k'$ distinct viable $\eps_{j-1}$-clusters whose pairwise distances are all greater than $\eps_j$.

Let us denote the first such layer by $l_1$, and the viable clusters by $X^{l_1}_1,\ldots,X^{l_1}_{k'}$ in increasing order of $A(\cdot)$ values. Note that any viable cluster $X_1$ must have a corresponding viable cluster $X'_1$ in the subsequent layer satisfying~\eqref{eq:cluster-step}. Moreover, for any two $\eps_{j-1}$-clusters $X_1,X_2$ in layer $L_{l_1}$ such that $A^+(X_1)<A^-(X_2)-\frac{\eps_j}{3}$ with flow into clusters $X'_1,X'_2$, respectively, in the subsequent layer, there cannot be any flow from $X_1$ to $X'_2$ or from $X_2$ to $X'_1$. Indeed, the distance from, say,  $X_1$ to $X'_2$ is greater than $\frac{\eps_j}3-k\eps_{j-1}>\eps_{j-1}\geq\eps_0$. In particular, for any layer containing $k'$ viable clusters with pairwise distance greater than $\frac{\eps_{j-1}}3$, each cluster must have flow into exactly one cluster in the subsequent layer (there cannot be more, since no layer contains more than $k'$ viable clusters in this phase).

Let us denote by $l_2$ the first layer after $l_1$ in which some pair of adjacent viable clusters $X^{l_2}_{i'},X^{l_2}_{i'+1}$ are at distance at most $\frac{\eps_j}3$. By the above argument, all viable flow in the portion of the graph from $L_{l_1}$ to $L_{l_2}$ flows through $k'$ disjoint cluster-paths $X^{l_1}_i\rightarrow X^{l_1+1}_i\rightarrow\ldots\rightarrow X^{l_2}_i$ (for $i=1,\ldots,k'$). Therefore, to reduce the number of viable clusters in all layers $L_{l_1},\ldots,L_{l_2}$, it suffices to cut just one cluster $X^l_i$ for some $i\in\{1,\ldots,k'\}$ and some $l_1<l<l_2$. Consider the pair of clusters $X^{l_2}_{i'},X^{l_2}_{i'+1}$. Since these clusters are at distance at most $\frac{\eps_j}3$ and the corresponding clusters $X^{l_1}_{i'},X^{l_1}_{i'+1}$ are at distance at least $\eps_j$, either $A^+(X^{l_2}_{i'})-A^+(X^{l_1}_{i'})\geq\frac{\eps_j}3$, or $A^-(X^{l_1}_{i'+1})-A^-(X^{l_2}_{i'+1})\geq\frac{\eps_j}3$. Without loss of generality, suppose the former.

Now consider the open real interval $(A^+(X^{l_1}_{i'}),A^+(X^{l_2}_{i'}))$. It has length at least $\frac{\eps_j}3$, and contains at most $k-1$ values in $\{A(v)\mid v\in L_{l_2}\}$ (at least one node $v$ in layer $L_{l_2}$ has $A(v)=A^+(X^{l_2}_{i'})$). Therefore there is an open subinterval $(a_0,a_1)$ of length at least $\frac{\eps_j}{3k}=4k\eps_{j-1}$ containing none of these values. By~\eqref{eq:cluster-step}, there must be some layer $L_l$ (for some $l_1<l<l_2$) for which $A^+(X^l_{i'})\in(a_0+2k\eps_{j-1},a_0+3k\eps_{j-1})$. Since $X^l_{i'}$ has width at most $k\eps_{j-1}$, it must also satisfy $A^-(X^l_{i'})\in(a_0+k\eps_{j-1},a_0+3k\eps_{j-1})$. In particular, all nodes in $X^l_{i'}$ are at distance at least $k\eps_{j-1}$ from all nodes in layer $L_{l_2}$. We can now repeat this argument for layers $>l_2$, until we have exhausted all layers in the graph, and by Lemma~\ref{lem:cut-charge}, the relative cost of the cut will be at most $1/(k\eps_{j-1})^2$, as required.
\end{proof}

%% file: lb.tex
In this section, we give a lower bound on the quality of approximation of algorithm $\SCRound$. It is not known whether this can be translated into an integrality gap for our LP (in fact, for our construction, the integrality gap is 1). We start by showing that the reduction to Markov flows discussed in Section~\ref{sec:flows} goes both ways. Specifically, we show the following lemma:

\begin{lemma}\label{lem:flows-to-gaps} Let $r>0$ be a positive integer, let $H=(L_0,\ldots,L_N,E)$ be a symmetric Markov flow graph with sources $L_0=\{s_0,s_1\}$ and sinks $L_N=\{t_0,t_1\}$ with at most $2^r$ states per layer, and let $F$ be a capacity-respecting flow from $s_0$ to $t_1$. Then there is a graph $G$ of pathwidth at most $2r-1$ with one demand pair $(s,t)$ and a feasible (but not necessarily optimal) solution $\{y_I\}$ to $\SCLP_{2r}(G)$ such that $H$ represents the distribution on assignments found by algorithm $\SCRound$ for $G$ and $\{y_I\}$, and $\tilde{y}_{s\neq t}\geq|F|$ (the amount of flow in $F$).
\end{lemma}
\begin{proof}  Let $F_{\mathrm{sym}}=\frac12(F+\bar{F})$ (where $\bar{F}$ is the flow from $s_1$ to $t_0$ corresponding symmetrically to $F$). By the symmetry of $H$, this is also a capacity-respecting (multi-)flow. Since the capacities in $H$ are themselves a flow, the residual capacity in $H$ can be decomposed into two multiflows $F_{\neq}$ and $F_=$ (between opposite and same-side terminals, respectively). As before, we may assume that both multiflows are symmetric. Thus, $H$ can be decomposed into a sum of flows $F_{\mathrm{sym}}+F_{\neq}+F_=$ from $\{s_0,s_1\}$ to $\{t_0,t_1\}$ where the total amount of flow between opposite terminals is $|F_{\mathrm{sym}}|+|F_{\neq}|\geq|F_{\mathrm{sym}}|=|F|$.

Now define a graph $G$ on vertices $\bigcup_{i=0}^{N}B'_i$, where $B'_0=\{s\}$, $B'_N=\{t\}$, and for all $0<i<N$, $|B'_i|=r$, and some edge set which admits a path-decomposition with bags $B_i=B'_i\cup B'_{i+1}$. Then for an appropriate symmetry-preserving correspondence between nodes of $H$ and local assignments to vertices of $G$, every path from $\{s_0,s_1\}$ to $\{t_0,t_1\}$ in $H$ corresponds to a full assignment $f:\bigcup B'_i\rightarrow\{0,1\}$. Thus the flow decomposition above can be viewed as a symmetric distribution on paths, which corresponds to a symmetric distribution on $\{0,1\}$ assignments in $G$. Let $\{\tilde{y}_I\}$ be the (level $n$) Sherali-Adams solution corresponding to this distribution. Note that $\tilde{y}_{s\neq t}=|F_{\mathrm{sym}}|+|F_{\neq}|\geq|F|$. It is also not hard to see that algorithm $\SCRound$ given the path decomposition and any scaling of $\{\tilde y_I\}$ will produce a distribution on assignments corresponding to the flow graph $H$. Thus, letting $\{y_I\}$ be an appropriate scaling (satisfying constraint~\eqref{LP:normalize}) completes the proof.
\end{proof}

By the above lemma, to get a lower-bound for our rounding which is exponential in the treewidth of the graph, it suffices, for every even integer $k\geq 4$, to construct a Markov flow graph as in Theorem~\ref{thm:flow-main} with at most $k$ nodes per layer, which admits $\Omega(k)(\frac12+A(t_1))$ units of capacity-respecting flow from $s_0$ to $t_1$ ($(\frac12+A(t_1))$ is the probability that the Markov chain starts and ends at opposite terminals $(s_0,t_1)$ or $(s_1,t_0)$ -- see Section~\ref{sec:potential} for the definition of $A(\cdot)$). Let us see such a construction now.

\paragraph{Construction} For every sufficiently large integer $N$ and $\eps\in(0,\frac1{2(N+k)})$, let $H_k(N,\eps)$ be the following layered capacitated digraph: As before, the nodes will consist of layers $L_0,L_1,\ldots,L_N$ where $L_0=\{s_0,s_1\}$ and $L_N=\{t_0,t_1\}$. For every $0<j<N$, we have $L_l=\{v^j_0,v^j_1,\ldots,v^j_{k-1}\}$.

For every $i=1,\ldots,k-2$ we add (directed) edges $(s_0,v^1_i)$ and $(s_1,v^1_{k-1-i})$ with capacities $2\eps(k-1-i)/(k-1)$ (respectively). We also add edges $(s_0,v^1_0)$ and $(s_1,v^1_{k-1})$ each with capacity $\frac12-(k-2)\eps$. Next, for between every two consecutive layers $L_j,L_{j+1}$ (for $1\leq j\leq N-2$) we add the following directed edges:
\begin{itemize}
   \item For all $i=1,\ldots,k-2$ add edges $(v^j_{i-1},v^{j+1}_i)$ and $(v^j_{i+1},v^{j+1}_i)$ each with capacity $\eps$.
   \item Add edges $(v^j_0,v^{j+1}_0)$ and $(v^j_{k-1},v^{j+1}_{k-1})$ each with capacity $\frac12-(j+k-2)\eps$.
   \item Add edges $(v^j_1,v^{j+1}_1)$ and $(v^j_{k-2},v^{j+1}_{k-2})$ each with capacity $j\eps$.
\end{itemize}
Finally, for $i=0,\ldots,\frac k2-1$ add edges $(v^{N-1}_i,t_0)$ and $(v^{N-1}_{k-1-i},t_1)$ with the full capacity of the respective layer $L_{N-1}$ node (i.e.\ capacity $2\eps$ for $i=2,\ldots,k-3$; capacity $\frac12-(N+k-4)\eps$ for $i=0,k-1$; and capacity $N\eps$ for $i=1,k-2$).

From the above construction, the definition of $A(\cdot)$, and Lemma~\ref{lem:A-stochastic}, the following claim follows immediately:
\begin{claim} In flow graph $H_k(N,\eps)$ we have
  \begin{enumerate}
    \item for all layers $j=1,\ldots N-1$ and all $i=0,\ldots,k$ we have $A(v^j_i)=\frac12-\frac i{k-1}$, and
    \item $A(t_0)=\frac12-\left(N+(\frac k2-2)(\frac k2+2)\right)\eps/k$ and $A(t_1)=-\frac12+\left(N+(\frac k2-2)(\frac k2+2)\right)\eps/(k-1)$.
  \end{enumerate}
\end{claim}
Thus, if we take $N=\omega(k^2)$ then \begin{equation}\label{eq:LB}\textstyle{\frac12}+A(t_1)=(1+o(1))N\eps/(k-1).\end{equation}

On the other hand, consider the flow (written as a weighted sum of paths) $F=\eps\sum_{j=1}^{N-k}p_j$, where path $p_j$ is defined as 
$$p_j=s_0\rightarrow v^1_0\rightarrow\ldots\rightarrow v^j_0\rightarrow v^{j+1}_1\rightarrow v^{j+2}_2\rightarrow\ldots\rightarrow v^{j+k-2}_{k-2}\rightarrow v^{j+k-1}_{k-2}\rightarrow\ldots\rightarrow v^{N-1}_{k-2}\rightarrow t_1.$$
It can readily be checked that $F$ is a capacity-respecting flow in $H_k(N,\eps)$ and that $|F|=(N-k)\eps$. Thus, by~\eqref{eq:LB}, for $N=\omega(k^2)$ we have $|F|\geq (1-o(1))(k-1)(\frac12+A(t_1))$, which is what we wanted to show.

%% file: hardness.tex
\section{NP-hardness for pathwidth 2}
\label{app:hardness}

\begin{theorem}\label{thm:hardness}
\SC (with general demands) is NP-hard even on graphs of pathwidth 2.
\end{theorem}
\begin{proof} The theorem follows from the following reduction from Max-Cut. Let graph $G=(V,E)$ be an instance of Max-Cut. Construct an instance of \SC on the graph $K_{2,n}$ as follows: identify every vertex $v_i\in V$ with a corresponding vertex $v'_i$ in the new graph. For every edge $(v_i,v_j)\in E$ add a demand pair $(v'_i,v'_j)$ with unit demand. Add two vertices $s,t$ and edges $\{(s,v'_i)\}_i$ and $\{(t,v'_i)\}_i$. Finally, make $(s,t)$ a demand pair with demand $n^3$. Consider some cut $(S,T)$ in the new graph. If $s\in S$ and $t\in T$ then the number of cut edges is exactly $n$. Thus the sparseness of the cut is exactly $n/(D(S,T))=n/(n^3+|E_G(S\setminus\{s\},T\setminus\{t\})|)$. Therefore the sparsest cut that separates $s$ from $t$ corresponds exactly to the max cut in $G$. It remains to show that the sparsest cut in the new graph must separate $s$ from $t$. Indeed, if $(S',T')$ is a cut for which $s,t\in S'$, then the sparseness of the cut is 
\begin{align*}\frac{|T'|}{D(S',T')}=\frac{|T'|}{|E_G(S'\setminus\{s,t\},T')|}&\geq\frac{n}{n|E_G(S'\setminus\{s,t\},T')|}\\&>\frac{n}{n^3+|E_G(S'\setminus\{s,t\},T')|}=\frac{n}{D(S'\setminus{t},T'\cap\{t\})}.
\end{align*}
\end{proof}
\iffalse
\begin{verbatim}
Yes exact is hard even for path-width 2 by the following reduction
from Max-Cut: Map the original vertices of the max-cut instance to an
independent set with one unit demand for every edge in the original
graph. Now connect all the vertices in the new graph to two new
vertices s,t (giving the graph K_{2,n}), and add, say n3 (or maybe we
need n4) demand between s and t (so that it never pays to put them on
the same side regardless of the other vertices). Then the number of
edges cut is always n, so the sparsest cut is the one with the most
demands cut, which is the max cut in the original graph.
\end{verbatim}
\fi

%%% Local Variables:
%%% mode: latex
%%% TeX-master: "BTWSC"
%%% End: